\documentclass[12pt]{article}

\usepackage[T1]{fontenc}
\usepackage[utf8]{inputenc}
\usepackage[a4paper,margin=1in]{geometry}
\usepackage{amsmath,amssymb,amsthm,mathtools}
\usepackage{microtype}
\usepackage{graphicx}
\usepackage{booktabs}
\usepackage{cite}
\usepackage{tikz}
\usetikzlibrary{arrows.meta,calc,positioning}
\usepackage[colorlinks=true,linkcolor=blue,citecolor=blue,urlcolor=blue]{hyperref}

\theoremstyle{plain}
\newtheorem{theorem}{Theorem}[section]
\newtheorem{proposition}[theorem]{Proposition}

\newtheorem{corollary}[theorem]{Corollary}

\theoremstyle{definition}
\newtheorem{definition}[theorem]{Definition}
\newtheorem{example}[theorem]{Example}

\theoremstyle{remark}
\newtheorem{remark}[theorem]{Remark}

\title{A Boundary--Residue Incidence Coalgebra for Associahedral Scattering Forms}

\author{Ioannis P.\ ZOIS\\
Exeter College,\\
Turl Street, OX1 3DP, Oxford, UK\\
\texttt{i.zois@exeter.oxon.org}}

\date{}

\begin{document}

\maketitle

\begin{abstract}
We propose a boundary--residue incidence coalgebra associated with the
face poset of a positive geometry and apply it to associahedral scattering
forms.  The construction is motivated by the analogy between the
Connes--Kreimer coproduct on Feynman graphs and the recursive residue
structure of canonical forms on positive geometries.  For the Stasheff
associahedron $K_n$, whose faces are indexed by non-crossing dissections
of an $(n+1)$-gon, we prove that the incidence coproduct records all
intermediate nested planar factorisation channels of the corresponding
tree-level scalar amplitude.  The residue of the canonical form on a face
labelled by a dissection factorises as the exterior product of canonical
forms on the lower associahedra associated with the resulting subpolygons.
The construction is illustrated explicitly for the pentagon associahedron
$K_4$, corresponding to the five-point planar scalar amplitude.  We then
formulate a loop-level extension: whenever a planar loop integrand is
represented by a positive geometry, the associahedral face poset is replaced
by the boundary poset of the corresponding loop geometry.  The one-loop
halohedron provides a concrete scalar example, while in the non-planar case
we define the associated incidence coalgebra at the level of logarithmic
singularity strata.  This gives a precise coalgebraic organisation of
boundary factorisation in positive geometry, parallel in spirit to the
Connes--Kreimer organisation of nested subdivergences in perturbative quantum
field theory.  We further compare the boundary--residue coalgebra with the
cellular incidence coalgebra of a triangulated or regular CW spacetime.  The
face poset of a finite regular CW complex reconstructs its barycentric
subdivision, and hence its underlying polyhedron, while in positive geometry
the same incidence mechanism organises canonical-form residues.  This yields a
precise incidence-first bridge between cellular spacetime topology and
positive-geometric amplitude factorisation, without assuming that metric or
causal data are determined by topology alone.
\end{abstract}

\section{Introduction}
\label{sec:introduction}

Perturbative quantum field theory traditionally organises scattering
amplitudes as sums of Feynman diagrams \cite{Feynman1949,tHooftVeltman1972,PeskinSchroeder1995,Weinberg1995,Dixon1996,ElvangHuang2015}.  This representation is extremely
successful, but it is not unique: individual diagrams depend on auxiliary
choices and are not themselves physical observables.  In many examples,
simple final amplitudes arise only after cancellations among many diagrammatic
terms.  This has motivated a search for formulations in which locality,
factorisation and unitarity are encoded more directly in the mathematical
objects from which amplitudes are extracted.

Several programmes in mathematical physics use non-diagrammatic or
non-standard mathematical structures to organise field-theoretic and
geometric data, from twistor methods to approaches to quantum gravity
\cite{Penrose1965,HawkingEllis1973,Polchinski1998,PenroseRindler1984,PenroseRindler1986,Rovelli2004,AshtekarLewandowski2004}.  In the present paper this broad context is used only as background.  The technical object studied here is much narrower: a precise incidence-coalgebraic structure inside the boundary and residue calculus of associahedral positive geometries.

The modern amplitudes programme provides one such formulation.  It is part of a broader development that includes twistor-string methods, on-shell recursion, colour--kinematics duality, the positive Grassmannian and the scattering-equation formalism \cite{Witten2004,BrittoCachazoFengWitten2005,BernCarrascoJohansson2008,CachazoHeYuan2014,ArkaniHamedPositiveGrassmannian2016}.  In a number
of theories, scattering amplitudes can be described by canonical differential
forms on positive geometries.  Two central examples are the amplituhedron,
which appears in planar maximally supersymmetric Yang--Mills theory, and the
associahedron, which appears in the positive geometry of kinematic space for
planar scalar theories \cite{ArkaniHamedTrnka2014,ArkaniHamedBaiLam2017,ArkaniHamedBaiHeYan2018,Mizera2017}.
In these settings, physical poles correspond to boundary components of the
geometry, and factorisation is implemented by residues of the canonical form.
Thus the boundary stratification of the positive geometry carries physical
information that is usually distributed across many Feynman diagrams.

A parallel algebraic lesson comes from the Connes--Kreimer approach to
renormalisation.  There, the combinatorics of Feynman graphs is organised by
a Hopf algebra whose coproduct records nested subdivergences and their
contractions \cite{Kreimer1998,ConnesKreimer1998,ConnesKreimer2000,ConnesKreimer2001,BlochKreimer2008}.  Although the associahedral amplitudes
considered in this article are tree-level objects and do not involve
renormalisation, they do possess a recursive boundary structure: faces of the
associahedron correspond to compatible factorisation channels, and iterated
residues of the canonical form correspond to nested factorisations.

The combinatorial background is equally important.  The associahedron was
introduced by Stasheff in the study of homotopy associativity, and is closely
related to Tamari orders, planar binary trees, operads, graph associahedra and
realizations of generalized associahedra \cite{Tamari1951,Stasheff1963,Lee1989,Loday2004,LodayRonco1998,Devadoss1999,ChapotonFominZelevinsky2002,CarrDevadoss2006,LodayVallette2012,Zois2005}.  Incidence coalgebras and Hopf algebras provide the natural algebraic language for intervals in posets and nested combinatorial structures \cite{Rota1964,JoniRota1979,Schmitt1994,AguiarMahajan2010,GrinbergReiner2014}.  We use this language to connect the combinatorics of associahedra with the residue recursion of canonical forms.

The main construction of this article is a boundary--residue incidence
coalgebra associated with the finite face poset of a positive geometry.  For the Stasheff associahedron $K_n$, whose faces are
indexed by non-crossing dissections of a polygon, we prove that this incidence
coproduct records all intermediate nested planar factorisation channels of the
corresponding tree-level scalar amplitude.  More explicitly, if a face is
labelled by a set of non-crossing diagonals, then the coproduct on the interval
from that face to the full associahedron sums over all intermediate
sub-dissections.  On canonical forms this is reflected by the factorisation of
residues into exterior products of the canonical forms of lower-dimensional
associahedra.

The construction is deliberately precise in scope.  It is not proposed as a
renormalisation Hopf algebra and, at tree level, it does not by itself compute
new amplitudes.  Rather, it isolates a coalgebraic structure that organises
the passage from boundary geometry to amplitude factorisation.  The first
non-trivial case, the pentagon associahedron $K_4$, is worked out explicitly
after the general statement for $K_n$.  The final part of the paper then
extends the same incidence-coalgebraic mechanism to loop positive geometries
when they are known, and formulates the non-planar obstruction in terms of a
logarithmic-singularity incidence coalgebra.

A further purpose of the paper is to make explicit a topological analogue of
this construction.  A triangulated or regular CW spacetime has a combinatorial
skeleton: cells, incidence relations, attaching data and skeleta.  In
algebraic topology, the face poset of a finite regular CW complex reconstructs
its barycentric subdivision, while orientation data determine the cellular
boundary operator.  Thus both the positive-geometric amplitude side and the
cellular-spacetime side pass through the same kind of incidence structure.
The paper formulates this as an incidence-first reconstruction principle:
continuum topological space and amplitude factorisation can both be organised
from combinatorial boundary data, although metric and causal structures require
additional labels.

\section{Polytopes and Associahedra: Basic Definitions}

\subsection{Polytopes}

We recall only the standard terminology needed below.  A convex polytope
can be specified either by its vertices or by a finite system of linear
inequalities.  These are the $V$- and $H$-representations of a polytope;
their equivalence is the Minkowski--Weyl theorem
\cite{Grunbaum2003,Ziegler1995}.\\

{\bf Definition 1.}
\emph{Convex hull description.}
A polytope is the convex hull of a finite set of points
$V=\{v_1,\ldots,v_k\}\subset\mathbb{R}^d$:
\[
P=\mathrm{conv}(V)
=\left\{\sum_{i=1}^k \lambda_i v_i \,\middle|\,
\lambda_i\ge 0,\; \sum_{i=1}^k \lambda_i=1\right\}.
\]

{\bf Definition 1'.} \emph{Half--space description.}
A polytope is a bounded intersection of finitely many half--spaces:
\[
P=\{x\in\mathbb{R}^d\mid A x \le b\},
\]
where $A$ is a real $m\times d$ matrix and $b$ is a real $m$-vector.\\

The combinatorial data used below are the faces of a polytope and their
incidence relations.  Vertices are $0$-faces, edges are $1$-faces, facets
are codimension-one faces, and inclusions of faces define the face poset.

\subsection{The Associahedron}

\subsubsection{Combinatorial motivation}

The associahedron encodes the combinatorics of full parenthesizations of
an ordered product, or equivalently planar binary trees \cite{Tamari1951,Stasheff1963,LodayRonco1998}.  Two vertices are
adjacent when the corresponding parenthesizations differ by one elementary
associativity move, for example
\[
((a\cdot b)\cdot c)\longmapsto (a\cdot (b\cdot c)).
\]
With the convention used here, the associahedron for $n$ input factors is
an $(n-2)$-dimensional polytope, and its number of vertices is the
corresponding Catalan number.

\begin{table}[htbp]
\centering
\begin{tabular}{c c c l}
\hline
$n$ (elements) & Parenthesizations & Dimension & Shape \\
\hline
2 & 1  & $0$ & Point \\
3 & 2  & $1$ & Line segment \\
4 & 5  & $2$ & Pentagon \\
5 & 14 & $3$ & Three-dimensional associahedron \\
6 & 42 & $4$ & Four-dimensional associahedron \\
\hline
\end{tabular}
\caption{Associahedra: number of parenthesizations, dimension, and shape.}
\label{tab:associahedra-catalan}
\end{table}

\subsubsection{Associahedra: definition, examples, and realisations}
The precise
mathematical definition is the following:\\

\begin{definition}[Associahedron]\label{def:associahedron}
The \emph{associahedron} $K_n$, also called the \emph{Stasheff polytope} \cite{Stasheff1963,Lee1989,Loday2004,ChapotonFominZelevinsky2002}, is a convex
polytope defined as follows. Let $n$ be the number of inputs (or leaves
in a tree). Then $K_n$ is a convex polytope of dimension $n-2$, whose:
\begin{itemize}
\item vertices correspond to the distinct full binary trees with $n$
leaves, or equivalently, the number of fully parenthesized products of
$n$ elements,
\item edges correspond to single associativity moves (tree rotations),
\item faces correspond to partial bracketings or degenerate trees.
\end{itemize}

Formally, let $T_n$ be the set of planar binary trees with $n$ leaves.
Then $K_n$ is a polytope with one-to-one correspondence between its:
\begin{itemize}
\item vertices and $T_n$,
\item $1$-skeleton (graph structure) and associativity relations
(rotation graph),
\item faces and combinatorial types of trees with missing internal
edges.
\end{itemize}
\end{definition}

\paragraph{Geometric realisation.}
There are several ways to realise the associahedron geometrically in
$\mathbb{R}^{n-2}$:
\begin{itemize}
\item Stasheff's construction via the poset of parenthesis expressions
(see \cite{Postnikov2009,Lam2016}).
\item Loday's realisation using tree coordinates \cite{Loday2004,LodayRonco1998}.
\item Secondary polytopes from convex hulls of certain points in space
(Gelfand--Kapranov--Zelevinsky construction, see \cite{GelfandKapranovZelevinsky1994}).
\end{itemize}

\paragraph{Some concrete examples.}
With the convention used here, $K_n$ corresponds to full parenthesizations
of a product of $n$ factors and has dimension $n-2$.
\begin{itemize}
\item $n=2$: $0$-dimensional associahedron. There is only one product
$a_1a_2$, so $K_2$ is a point.
\item $n=3$: $1$-dimensional associahedron. There are two full
parenthesizations,
\[
(a_1a_2)a_3,\qquad a_1(a_2a_3),
\]
so $K_3$ is a line segment.
\item $n=4$: $2$-dimensional associahedron. There are five full
parenthesizations of $a_1a_2a_3a_4$:
\[
((a_1a_2)a_3)a_4,\quad
(a_1(a_2a_3))a_4,\quad
a_1((a_2a_3)a_4),\quad
a_1(a_2(a_3a_4)),\quad
(a_1a_2)(a_3a_4),
\]
so $K_4$ is a pentagon.
\item $n=5$: $3$-dimensional associahedron. There are $14$ full
parenthesizations of $a_1a_2a_3a_4a_5$, so $K_5$ is a three-dimensional
polytope with $14$ vertices.
\end{itemize}

One standard geometric model of the associahedron uses triangulations of a
convex polygon.

Consider a convex $(n+1)$--gon. A triangulation of this polygon is a
decomposition into triangles by drawing non--intersecting diagonals.
Each triangulation uses exactly $n-2$ diagonals and there are
$C_{n-1}$ such triangulations, where $C_n$ denotes the $n$--th Catalan
number.

There is a natural correspondence between triangulations of the
$(n+1)$--gon and the vertices of the associahedron $K_n$. Each diagonal
corresponds to a partial parenthesization and each complete
triangulation corresponds to a full parenthesization.

Two triangulations are related by a single diagonal flip if and only if
the corresponding vertices of the associahedron are connected by an
edge.

This polygon picture provides a concrete geometric model for the
associahedron and makes its combinatorial structure transparent.

Associahedra occur in several areas of mathematics and mathematical physics:
\begin{itemize}
\item Algebraic topology: models associativity in loop spaces.
\item Category theory: higher associativity and operads (see \cite{Stasheff1963,Loday2004}).
\item Mathematical physics: scattering amplitudes and positive geometries.
\end{itemize}

\section{Associahedra, Scattering Amplitudes and Feynman Graphs}

We shall focus on the third application: scattering amplitudes in quantum
field theory (see \cite{Devadoss1999,Devadoss2001,GelfandKapranovZelevinsky1994,Brown2009}). The
aim is to reformulate the computation of scattering amplitudes --- how
particles interact and scatter after collisions --- in a way that is
geometrically natural, combinatorially clean, and avoids the
redundancies of Feynman diagrams.
The appearance of associahedra in scattering amplitudes is clearest in
planar scalar $\phi^3$ theory.

Traditionally, scattering amplitudes are computed by summing over
Feynman diagrams. For an $n$--particle tree--level process in $\phi^3$
theory, each planar Feynman diagram corresponds precisely to a planar
binary tree, and hence to a vertex of the associahedron.

Each propagator in a Feynman diagram corresponds to a diagonal in the
polygon picture and therefore to a facet of the associahedron.

This observation leads to a geometric representation of the colour-ordered
tree amplitude as the canonical differential form associated with the
kinematic associahedron.

In this formulation, locality and unitarity arise from the boundary
structure of the polytope itself.
In scalar field theories, especially bi--adjoint scalar theory, it was
discovered that the kinematic data of particle scattering can be
embedded into a space where the amplitudes live on the boundaries of a
polytope --- the associahedron. In greater detail, what has been shown is
that scattering amplitudes (in certain quantum field theories) can be
calculated as volumes of associahedra in a certain \emph{kinematic
space}. The boundaries of the associahedron correspond to factorization
channels in the scattering process, when a process splits into
sub--processes.

For the class of amplitudes considered here, this gives a formulation in
which the final result is extracted from the geometry of the associahedron
rather than from an explicit summation over planar cubic diagrams.

The relevant dictionary is the following:
\begin{itemize}
\item vertices represent distinct orderings and ways particles can
interact (parentheses or diagrams),
\item edges represent possible transitions (factorization limits),
\item facets correspond to physical poles in the amplitude (when
particles become collinear or on--shell).
\end{itemize}

The advantage of this description is that the pole and factorisation
structure of the amplitude is controlled by the boundary structure of a
convex polytope.  In what follows we use this boundary structure, rather
than computational efficiency, as the basis for the coalgebraic construction.

We now examine this construction in greater detail. We restrict attention
first to the tree-level planar sector of massless scalar $\phi^3$ theory.
Let $p_1,\ldots,p_n$ be massless external momenta in four-dimensional
Minkowski space, satisfying
\begin{align}
\sum_{i=1}^n p_i &= 0 , \label{eq:momentum-conservation} \\
p_i^2 &= 0 , \qquad i=1,\ldots,n . \label{eq:on-shell}
\end{align}
For a fixed cyclic ordering $(1,2,\ldots,n)$, define the planar kinematic
variables
\[
X_{ij}=(p_i+p_{i+1}+\cdots+p_{j-1})^2,
\qquad 1\leq i<j\leq n,
\]
where the indices are understood cyclically and the adjacent cases are
excluded. These variables are attached to diagonals of an $n$-gon; hence
there are $n(n-3)/2$ planar variables.

\begin{definition}[Planar kinematic space]
The planar kinematic space $\mathcal{K}_n^{\mathrm{pl}}$ is the real vector
space with coordinates $X_{ij}$ corresponding to the diagonals of a
fixed cyclically ordered $n$-gon. Thus
\[
\dim \mathcal{K}_n^{\mathrm{pl}}=\frac{n(n-3)}{2}.
\]
\end{definition}

\begin{definition}[Kinematic associahedron]
The kinematic associahedron is obtained by intersecting the positive
region
\[
X_{ij}\geq 0
\]
with a suitable affine subspace of $\mathcal{K}_n^{\mathrm{pl}}$. In the
standard construction of Arkani-Hamed, Bai, He and Yan, this affine
subspace is fixed by assigning positive constants to the non-planar
Mandelstam combinations. The resulting bounded polytope is combinatorially
the associahedron: its facets are the physical poles $X_{ij}=0$, and its
vertices correspond to planar cubic Feynman diagrams.
\end{definition}

\begin{definition}[Canonical form]
Let $P$ be a positive geometry, for example a convex polytope. Its
canonical form $\Omega(P)$ is the unique meromorphic top-degree form with
logarithmic singularities on the boundary hypersurfaces and residues equal
to the canonical forms of the boundary geometries. For the kinematic
associahedron one writes schematically
\[
\Omega(A_n)=m_n\, d^{n-3}X,
\]
where the rational function $m_n$ is the corresponding colour-ordered
tree-level scalar amplitude.
\end{definition}

As explicit examples, we shall compute the four- and five-particle
scattering amplitudes in cubic scalar theory both via Feynman graphs and
via the associahedron method, and verify that the two approaches coincide.

\subsection[Four-Particle Scattering Amplitude]{Four--Particle Scattering Amplitude ($\phi^3$ theory)}

For the full non-colour-ordered four-point tree amplitude in cubic scalar
theory, the usual Feynman-diagram computation gives three channels:
\[
A_4^{\mathrm{full}}=\frac{1}{s}+\frac{1}{t}+\frac{1}{u},
\]
where, for massless external momenta,
\[
s=(p_1+p_2)^2,\qquad t=(p_1+p_3)^2,\qquad u=(p_1+p_4)^2 .
\]

For a fixed planar ordering, however, only the planar channels appear. For
example, with cyclic ordering $(1,2,3,4)$ the two planar variables may be
chosen as
\[
X_{13}=s_{12}=(p_1+p_2)^2,\qquad
X_{24}=s_{23}=(p_2+p_3)^2.
\]
The four-point associahedron is a line segment. Its two endpoints are the
two planar factorisation channels. On the affine line
\[
X_{13}+X_{24}=C,\qquad C>0,
\]
with $X_{13},X_{24}>0$, the canonical logarithmic form is
\[
\Omega_4
=d\log\left(\frac{X_{13}}{X_{24}}\right)
=\left(\frac{1}{X_{13}}+\frac{1}{X_{24}}\right)dX_{13}.
\]
Thus the colour-ordered planar scalar amplitude is
\[
m_4(1,2,3,4)=\frac{1}{X_{13}}+\frac{1}{X_{24}}.
\]
The full non-colour-ordered expression is recovered only after summing the
appropriate planar contributions over inequivalent cyclic orderings; it is
not a single one-dimensional associahedron.

\begin{figure}[htbp]
\centering
\begin{tikzpicture}[scale=1.1]
  \draw[very thick] (0,0)--(5,0);
  \fill (0,0) circle (2pt) node[below] {$X_{13}=0$};
  \fill (5,0) circle (2pt) node[below] {$X_{24}=0$};
  \node[above] at (2.5,0) {$K_3$ / four-point planar associahedron};
\end{tikzpicture}
\caption{The four-point planar associahedron is a line segment. Its endpoints
represent the two planar factorisation channels.}
\label{fig:four-point-associahedron}
\end{figure}

\subsection[Five-Particle Scattering Amplitude]{Five-Particle Scattering Amplitude ($\phi^3$ theory)}

For five massless particles $(p_1,p_2,p_3,p_4,p_5)$ with planar ordering
$(1,2,3,4,5)$, the associahedron is a two-dimensional pentagon. The five
planar variables corresponding to diagonals of the pentagon are
\begin{align*}
X_{13} &= s_{12}=(p_1+p_2)^2, \\
X_{14} &= s_{123}=(p_1+p_2+p_3)^2, \\
X_{24} &= s_{23}=(p_2+p_3)^2, \\
X_{25} &= s_{234}=(p_2+p_3+p_4)^2, \\
X_{35} &= s_{34}=(p_3+p_4)^2 .
\end{align*}

Each vertex of the pentagon corresponds to a triangulation of the pentagon,
or equivalently to a planar cubic tree. The compatible pairs of diagonals
are
\[
(13,14),\quad (14,24),\quad (24,25),\quad (25,35),\quad (13,35).
\]
Therefore the colour-ordered five-point tree amplitude is
\[
\begin{aligned}
m_5(1,2,3,4,5)
={}&\frac{1}{X_{13}X_{14}}
+\frac{1}{X_{14}X_{24}}
+\frac{1}{X_{24}X_{25}} \\
&+\frac{1}{X_{25}X_{35}}
+\frac{1}{X_{13}X_{35}} .
\end{aligned}
\]
This is the canonical rational function associated with the pentagon
associahedron. Its five boundary facets encode the five factorisation
channels of the planar amplitude.

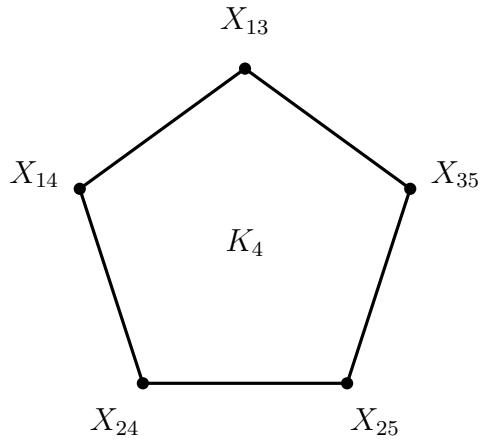
\begin{figure}[htbp]
\centering
\begin{tikzpicture}[scale=1.15]
  \foreach \i in {1,...,5} {\coordinate (v\i) at ({90+72*(\i-1)}:2);}
  \draw[very thick] (v1)--(v2)--(v3)--(v4)--(v5)--cycle;
  \foreach \i/\lab in {1/{$X_{13}$},2/{$X_{14}$},3/{$X_{24}$},4/{$X_{25}$},5/{$X_{35}$}} {
    \fill (v\i) circle (2pt);
    \node at ({90+72*(\i-1)}:2.55) {\lab};
  }
  \node at (0,0) {$K_4$};
\end{tikzpicture}
\caption{The five-point planar associahedron is a pentagon. Vertices correspond
to triangulations, or equivalently to planar cubic trees.}
\label{fig:five-point-associahedron}
\end{figure}

\section{Amplituhedra and Positive Geometries}

\begin{definition}[Amplituhedron]
Let $k,n,m$ be integers with $k+m\leq n$. Let
\[
Z=(Z_1,\ldots,Z_n)
\]
be a totally positive configuration of $n$ vectors in $\mathbb{R}^{k+m}$;
equivalently, after arranging the $Z_i$ as the rows of an
$n\times(k+m)$ matrix, all ordered maximal minors are positive. Let
$G_+(k,n)\subset \mathrm{Gr}(k,n)$ denote the positive Grassmannian.
For $C\in G_+(k,n)$, set
\[
Y=CZ\in \mathrm{Mat}_{k\times(k+m)}(\mathbb{R}),
\]
and regard the row span of $Y$ as a point of $\mathrm{Gr}(k,k+m)$. The
amplituhedron is the image
\[
\mathcal{A}_{n,k,m}(Z)
=\{\,Y=CZ\in \mathrm{Gr}(k,k+m)\mid C\in G_+(k,n)\,\}.
\]
\end{definition}

\begin{definition}[Positive geometry]
Following the formalism of positive geometries and canonical forms \cite{ArkaniHamedBaiLam2017,Postnikov2006,SpeyerWilliams2005,Lam2016,Scott2006}, a positive geometry is, roughly, a pair $(X,X_{\geq0})$, where $X$ is a
complex projective variety and $X_{\geq0}\subset X(\mathbb{R})$ is an
oriented real semi-algebraic subset with boundary, together with a unique
canonical differential form $\Omega(X,X_{\geq0})$ having logarithmic
singularities on all boundary components and no other singularities.
Its residues on boundary components are the canonical forms of the
corresponding boundary geometries.
\end{definition}

This canonical form encodes geometric and physical information. In the
amplituhedron programme, scattering amplitudes in planar
$\mathcal{N}=4$ supersymmetric Yang--Mills theory are obtained from the
canonical form of the amplituhedron. For $k=1$ the image is closely related
to a cyclic polytope. The associahedron is a different positive geometry,
most directly associated with planar cubic scalar amplitudes in kinematic
space.

As a concrete example related to positive geometry, recall the
Parke--Taylor formula. The Parke--Taylor formula gives a compact expression
for colour-ordered tree-level maximally helicity violating (MHV) gluon
amplitudes in Yang--Mills theory \cite{ParkeTaylor1986,Dixon1996,ElvangHuang2015}. These
are amplitudes for which exactly two gluons have negative helicity and all
remaining gluons have positive helicity.

Let the external gluons be labelled $1,2,\ldots,n$, and let gluons $i$ and
$j$ be the only negative-helicity particles. The colour-ordered partial
amplitude is
\[
A_n^{\mathrm{tree}}
(1^+,\ldots,i^-,\ldots,j^-,\ldots,n^+)
=
\frac{\langle ij\rangle^4}
{\langle12\rangle\langle23\rangle\cdots\langle n1\rangle} .
\]
Here
\[
\langle ab\rangle
=
\varepsilon_{\alpha\beta}\lambda_a^{\alpha}\lambda_b^{\beta}
\]
denotes the usual spinor-helicity bracket for massless momenta. For example,
\[
A^{\mathrm{MHV}}_5(1^-,2^-,3^+,4^+,5^+)
=
\frac{\langle12\rangle^4}
{\langle12\rangle\langle23\rangle\langle34\rangle
\langle45\rangle\langle51\rangle}.
\]

It is important not to conflate two related but distinct statements. The
factor
\[
\mathrm{PT}(1,2,\ldots,n)=
\frac{1}{\langle12\rangle\langle23\rangle\cdots\langle n1\rangle}
\]
is the Parke--Taylor denominator appearing in the gluon MHV amplitude. On
the other hand, the positive moduli space $\mathcal{M}_{0,n}^{+}$ carries
canonical logarithmic forms that are closely related to Parke--Taylor forms
and to the CHY/twistor-string representations of amplitudes \cite{Witten2004,CachazoHeYuan2014,CachazoHeYuan2014II,Mizera2017}. Thus positive
geometry explains the structure of these forms, but the full Yang--Mills
MHV amplitude also includes the helicity numerator $\langle ij\rangle^4$
and the corresponding momentum-conservation data.

For $n=5$, after fixing three points by $\mathrm{SL}(2)$ invariance, one may
use coordinates $0<t_2<t_3<1$ on a positive cell of $\mathcal{M}_{0,5}^+$.
A representative logarithmic form is
\[
\omega_5
=
\frac{dt_2\wedge dt_3}
{t_2(1-t_2)t_3(1-t_3)(t_3-t_2)}.
\]
This form should be understood as a positive-geometric ingredient in the
amplitudes programme, rather than as a direct elementary derivation of the
spinor-helicity Parke--Taylor formula by itself.

\section{Boundary Factorisation from Positive Geometry}
\label{sec:boundary-factorisation}

In this section we isolate the technical contrast that will be used in the
coalgebraic construction.  No foundational claim is needed here.  The narrower
point is that, for the amplitudes under consideration, some information normally
represented by Feynman diagrams is encoded instead in the boundary and residue
structure of a positive geometry.

In a diagrammatic formulation, a tree amplitude is assembled from local
vertices and propagators.  For planar cubic scalar theory these propagators
are labelled by planar channels, and the allowed compatible collections of
channels are precisely the non-crossing dissections of a polygon.  The
associahedron packages these compatible collections into a single positive
geometry.  Its facets correspond to physical poles, while residues on these
facets give the lower-point canonical forms associated with factorisation.

The same point can be expressed algebraically.  The canonical form of a
positive geometry satisfies a recursive boundary condition: the residue on a
boundary component is the canonical form of that boundary component.  For the
associahedron, this recursion is the recursion of planar factorisation.  Thus,
at tree level, locality is not separately imposed at the level of individual
diagrams; it is encoded by the face structure and residue structure of the
positive geometry.

The operational contrast relevant for this paper is summarised in
Table~\ref{tab:geometry-qft}.  The table should be read as a comparison of
organising principles, not as an assertion that the two formalisms are
identical in scope.

\begin{table}[htbp]
\centering
\begin{tabular}{l l}
\hline
Diagrammatic organisation & Positive-geometric organisation \\
\hline
Planar Feynman diagrams & Vertices and faces of the associahedron \\
Propagator channels & Facets / boundary hypersurfaces \\
Factorisation limits & Residues of the canonical form \\
Nested compatible channels & Chains in the face poset \\
Diagrammatic recursion & Boundary--residue recursion \\
\hline
\end{tabular}
\caption{Operational contrast between the diagrammatic and
positive-geometric organisations of tree-level planar amplitudes.}
\label{tab:geometry-qft}
\end{table}

For the amplitudes considered here, the relevant structural data are encoded
in a positive geometry together with its canonical form.  The pole and
residue structure of that form records the same factorisation information
that is represented diagrammatically by propagator channels.  The
boundary--residue incidence coalgebra introduced below records this recursive
structure at the level of the face poset.

\section{A Boundary--Residue Coalgebra for Positive Geometries}
\label{sec:boundary-residue-coalgebra}

The analogy with the Connes--Kreimer Hopf algebra suggests that the
boundary stratification of a positive geometry should itself carry an
algebraic structure.  We now make this idea precise in the simplest
setting, namely for positive geometries whose boundary strata form a
finite face poset, such as convex polytopes and, in particular, the
associahedron.  The construction is an incidence coalgebra in the sense
of Rota's theory of incidence algebras and posets, adapted to the
boundary strata of positive geometries \cite{Rota1964,JoniRota1979,Schmitt1994,StanleyEC1,StanleyEC2}.

Let $P$ be a positive geometry with finite face poset $\mathcal F(P)$,
ordered by inclusion. Thus $F\leq G$ means that $F$ is a face of $G$.
For $F\leq G$, denote by $[F,G]$ the corresponding interval in the face
poset.

\begin{definition}[Boundary--residue incidence coalgebra]
Let $\mathcal C(P)$ be the vector space over $\mathbb Q$ generated by all
intervals $[F,G]$, with $F\leq G$ in $\mathcal F(P)$. Define a coproduct
\[
   \Delta [F,G]
   =
   \sum_{F\leq H\leq G} [F,H]\otimes [H,G],
\]
and a counit
\[
   \varepsilon([F,G]) =
   \begin{cases}
   1, & F=G,\\
   0, & F<G.
   \end{cases}
\]
We call $\mathcal C(P)$ the boundary--residue incidence coalgebra of $P$.
\end{definition}

\begin{proposition}
The coproduct $\Delta$ is coassociative.
\end{proposition}

\begin{proof}
For any interval $[F,G]$, we compute
\[
(\Delta\otimes \mathrm{id})\Delta[F,G]
 =
 \sum_{F\leq H\leq K\leq G}
 [F,H]\otimes[H,K]\otimes[K,G],
\]
whereas
\[
(\mathrm{id}\otimes \Delta)\Delta[F,G]
 =
 \sum_{F\leq H\leq K\leq G}
 [F,H]\otimes[H,K]\otimes[K,G].
\]
The two expressions are identical. Hence $\Delta$ is coassociative.
\end{proof}

The relevance of this coalgebra to positive geometry comes from the
recursive definition of canonical forms. If $\Omega_P$ denotes the
canonical form of $P$, then for a boundary face $F\subset P$ one has
\[
   \operatorname{Res}_{F}\Omega_P = \Omega_F,
\]
up to the conventional orientation sign. More generally, for a chain of
faces
\[
   F_0\leq F_1\leq \cdots \leq F_r=P,
\]
iterated residues satisfy
\[
   \operatorname{Res}_{F_0}
   \operatorname{Res}_{F_1}
   \cdots
   \operatorname{Res}_{F_{r-1}}
   \Omega_P
   =
   \Omega_{F_0},
\]
again up to orientation. This is the residue recursion built into the
definition of positive geometries and canonical forms
\cite{ArkaniHamedBaiLam2017}.

\begin{proposition}[Compatibility of residues with the boundary coalgebra]
Let $P$ be a positive geometry with finite face poset. The coproduct on
$\mathcal C(P)$ records all possible intermediate boundary factorisations
of the residue operation. In particular, the terms $[F,H]\otimes[H,P]$
appearing in
\[
   \Delta[F,P]
   =
   \sum_{F\leq H\leq P}
   [F,H]\otimes[H,P]
\]
are in one-to-one correspondence with the possible ways of taking the
residue of $\Omega_P$ first along an intermediate face $H$, and then
along $F\subset H$.
\end{proposition}

\begin{proof}
By the defining property of a positive geometry, taking a residue along a
boundary face $H\subset P$ gives the canonical form $\Omega_H$. Taking a
further residue along $F\subset H$ gives $\Omega_F$. Thus each
intermediate face $H$ determines a factorisation of the residue operation
\[
   \Omega_P \longrightarrow \Omega_H \longrightarrow \Omega_F.
\]
The coproduct above sums precisely over all such intermediate faces $H$
with $F\leq H\leq P$. Hence the incidence coproduct is the combinatorial
shadow of the recursive residue structure of canonical forms.
\end{proof}

This construction is deliberately parallel to the Connes--Kreimer
coproduct
\[
   \Delta(\Gamma)
   =
   \sum_{\gamma\subset \Gamma}
   \gamma\otimes \Gamma/\gamma,
\]
where the sum runs over divergent subgraphs of a Feynman graph $\Gamma$
\cite{Kreimer1998,ConnesKreimer1998,ConnesKreimer2000,ConnesKreimer2001}. In the Connes--Kreimer case, the coproduct
records the nested subdivergences of perturbative renormalisation. In the
positive geometry case, the boundary--residue incidence coproduct records the
nested factorisation channels of the canonical form.

For the associahedron this interpretation becomes particularly concrete.
Let $A_n$ denote the kinematic associahedron associated with an
$n$-point planar tree amplitude; with the convention used earlier, this
has the combinatorial type of the Stasheff polytope $K_{n-1}$. A face of
$A_n$ is specified by a set $D$ of mutually non-crossing diagonals of an
$n$-gon. Cutting the polygon along these diagonals decomposes it into
smaller polygons. Correspondingly, the face is combinatorially a product
of lower-dimensional associahedra:
\[
   F_D \simeq \prod_{\alpha} A_{n_\alpha},
\]
where the $n_\alpha$ are the numbers of external sides of the resulting
subpolygons, with the standard identifications along the internal
diagonals. Consequently, the residue of the canonical form on such a face
factorises as
\[
   \operatorname{Res}_{F_D}\Omega_{A_n}
   =
   \bigwedge_{\alpha}\Omega_{A_{n_\alpha}},
\]
up to orientation and the standard momentum-conservation constraints.
This is precisely the geometric form of tree-level factorisation: a pole
corresponding to a planar channel decomposes the amplitude into
lower-point amplitudes \cite{ArkaniHamedBaiHeYan2018,Mizera2017}.

\begin{corollary}
For the kinematic associahedron $A_n$, the boundary--residue incidence coalgebra
$\mathcal C(A_n)$ encodes the nested planar factorisation structure of
tree-level scalar amplitudes.
\end{corollary}

\subsection[The general associahedron Kn]{The general associahedron $K_n$}
\label{subsec:general-Kn}

We now formulate the preceding statement in the standard combinatorial
language of the Stasheff associahedron.  In this subsection, $K_n$ denotes
the associahedron whose vertices are the triangulations of a convex
$(n+1)$-gon.  Thus $K_4$ is the pentagon, corresponding to the five-point
planar scalar amplitude in the convention used in this article.

Let $P_{n+1}$ be a convex polygon with vertices labelled cyclically by
$1,2,\ldots,n+1$.  A set $D$ of diagonals of $P_{n+1}$ will be called
admissible if its elements are pairwise non-crossing.  To every admissible
set $D$ there corresponds a face $F_D\subset K_n$.  The incidence order is
reversed with respect to inclusion of diagonals:
\[
   F_D\subseteq F_E \quad \Longleftrightarrow \quad E\subseteq D .
\]
In particular, facets correspond to single diagonals, while vertices
correspond to maximal admissible sets, i.e. triangulations.

\begin{theorem}[Boundary--residue coalgebra of $K_n$]
\label{thm:Kn-boundary-residue}
Let $K_n$ be the associahedron of a convex $(n+1)$-gon, and let
$D$ be an admissible set of non-crossing diagonals.  Then the interval
coproduct from the face $F_D$ to the full associahedron is
\[
   \Delta [F_D,K_n]
   =
   \sum_{E\subseteq D}
   [F_D,F_E]\otimes [F_E,K_n],
\]
where the sum runs over all sub-dissections $E$ of $D$.  Moreover, if the
cutting of $P_{n+1}$ along $D$ produces subpolygons $Q_1,\ldots,Q_r$, and
if $q_\alpha$ is the number of sides of $Q_\alpha$, then
\[
   F_D \simeq \prod_{\alpha=1}^{r} K_{q_\alpha-1}
\]
combinatorially, with the convention that the associahedron of a triangle
is a point.  Consequently the canonical form satisfies the residue
factorisation
\[
   \operatorname{Res}_{F_D}\Omega_{K_n}
   =
   \pm\bigwedge_{\alpha=1}^{r}\Omega_{K_{q_\alpha-1}},
\]
where the sign depends only on orientation conventions.  The coproduct
therefore records all possible first-stage choices of factorisation
channels leading from the full associahedron to the boundary stratum
$F_D$.
\end{theorem}

\begin{proof}
The face structure of the associahedron is controlled by polygon
dissections.  A face is determined by specifying a collection $D$ of
pairwise non-crossing diagonals.  Adding further compatible diagonals
refines the dissection and lowers the dimension of the face; hence face
inclusion is reversed with respect to inclusion of diagonal sets.  Thus,
for a face $H$ satisfying
\[
   F_D\subseteq H\subseteq K_n,
\]
there is a unique admissible set $E$ with $E\subseteq D$ such that
$H=F_E$.  Conversely, every subset $E\subseteq D$ is again admissible and
defines an intermediate face $F_E$ with
\[
   F_D\subseteq F_E\subseteq K_n.
\]
Applying the definition of the incidence coproduct to the interval
$[F_D,K_n]$ therefore gives
\[
   \Delta [F_D,K_n]
   =
   \sum_{F_D\leq H\leq K_n}[F_D,H]\otimes[H,K_n]
   =
   \sum_{E\subseteq D}[F_D,F_E]\otimes[F_E,K_n].
\]

It remains to identify the residue.  Cutting the polygon $P_{n+1}$ along
all diagonals in $D$ decomposes it into subpolygons
$Q_1,\ldots,Q_r$.  The standard product structure of faces of the
associahedron gives a combinatorial identification
\[
   F_D \simeq \prod_{\alpha=1}^{r} K_{q_\alpha-1},
\]
where $q_\alpha$ is the number of sides of $Q_\alpha$.  On the level of
positive geometries, the canonical form on a product is the exterior
product of the canonical forms of the factors, and the defining recursive
property of canonical forms states that the residue of $\Omega_{K_n}$ on
a boundary face is the canonical form of that face.  Hence
\[
   \operatorname{Res}_{F_D}\Omega_{K_n}
   =
   \pm \Omega_{F_D}
   =
   \pm\bigwedge_{\alpha=1}^{r}\Omega_{K_{q_\alpha-1}}.
\]
This is precisely the positive-geometric form of planar tree-level
factorisation: the diagonals in $D$ label the propagator channels, and the
resulting subpolygons label the lower-point amplitudes appearing on the
boundary.  The coproduct sums over all sub-dissections $E\subseteq D$, and
therefore over all possible first-stage choices of intermediate
factorisation channels.  This proves the claim.
\end{proof}

\begin{remark}
The theorem should be read as a structural statement about the
associahedral organisation of tree-level planar factorisation.  The
incidence coalgebra itself is the classical coalgebra of intervals in a
finite poset; the additional content here is its interpretation through
canonical forms and their residues on positive geometries.  In this sense,
the construction is not a new incidence coalgebra in isolation, but a
boundary--residue realisation of incidence coalgebra in the positive
geometry of scattering amplitudes.
\end{remark}

\subsection[Worked example: the pentagon associahedron K4]{Worked example: the pentagon associahedron $K_4$}
\label{subsec:K4-boundary-coalgebra}

We now spell out the preceding construction in the first non-trivial
case.  With the convention used in this article, $K_4$ is the pentagon
associahedron.  It has five vertices, five edges, and one two-dimensional
cell.  Combinatorially, it is the associahedron associated with the five
planar cubic diagrams of the five-point scalar amplitude.

Let the five facets of $K_4$ be labelled by the planar kinematic channels
\[
   E_{13},\quad E_{14},\quad E_{24},\quad E_{25},\quad E_{35}.
\]
The five vertices are the intersections of compatible pairs of adjacent
facets:
\[
\begin{aligned}
   V_1 &= E_{13}\cap E_{14}, &
   V_2 &= E_{14}\cap E_{24}, &
   V_3 &= E_{24}\cap E_{25}, \\
   V_4 &= E_{25}\cap E_{35}, &
   V_5 &= E_{35}\cap E_{13}. &
\end{aligned}
\]
Equivalently, the vertices correspond to the five terms in the planar
five-point scalar amplitude,
\[
   m_5^{\mathrm{planar}}
   =
   \frac{1}{X_{13}X_{14}}
   +\frac{1}{X_{14}X_{24}}
   +\frac{1}{X_{24}X_{25}}
   +\frac{1}{X_{25}X_{35}}
   +\frac{1}{X_{35}X_{13}}.
\]

\begin{figure}[h]
\centering
\begin{tikzpicture}[scale=2.0]
  \coordinate (V1) at (90:1.35);
  \coordinate (V2) at (18:1.35);
  \coordinate (V3) at (-54:1.35);
  \coordinate (V4) at (-126:1.35);
  \coordinate (V5) at (162:1.35);
  \draw[thick] (V1)--(V2)--(V3)--(V4)--(V5)--cycle;
  \fill (V1) circle (0.035) node[above] {$V_1$};
  \fill (V2) circle (0.035) node[right] {$V_2$};
  \fill (V3) circle (0.035) node[below right] {$V_3$};
  \fill (V4) circle (0.035) node[below left] {$V_4$};
  \fill (V5) circle (0.035) node[left] {$V_5$};
  \node at ($(V1)!0.5!(V2)+(0.16,0.12)$) {$E_{14}$};
  \node at ($(V2)!0.5!(V3)+(0.18,-0.02)$) {$E_{24}$};
  \node at ($(V3)!0.5!(V4)+(0,-0.18)$) {$E_{25}$};
  \node at ($(V4)!0.5!(V5)+(-0.18,-0.02)$) {$E_{35}$};
  \node at ($(V5)!0.5!(V1)+(-0.16,0.12)$) {$E_{13}$};
  \node at (0,0) {$K_4$};
\end{tikzpicture}
\caption{The pentagon associahedron $K_4$. Its five edges are the five
planar channels of the five-point scalar amplitude, and its five vertices
are the five compatible pairs of channels.}
\label{fig:K4-pentagon-coalgebra}
\end{figure}
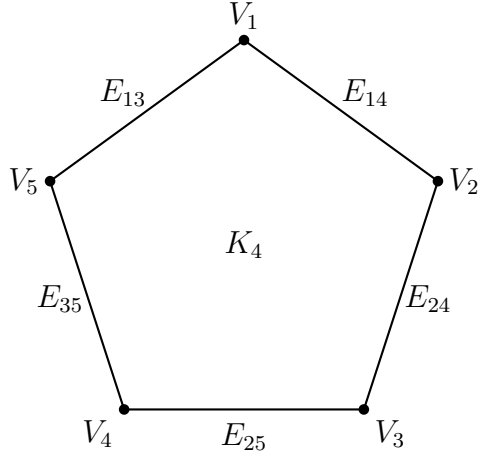

The non-empty face poset of $K_4$ is therefore
\[
   \mathcal F(K_4)
   =
   \{K_4\}\cup \{E_{13},E_{14},E_{24},E_{25},E_{35}\}
   \cup \{V_1,V_2,V_3,V_4,V_5\},
\]
with incidence relations determined by the pentagon.  For instance,
$V_1$ lies on $E_{13}$ and $E_{14}$, while $V_2$ lies on $E_{14}$ and
$E_{24}$.

For a facet $E\subset K_4$, the boundary--residue coproduct gives
\[
   \Delta[E,K_4]
   =
   [E,E]\otimes[E,K_4]
   +
   [E,K_4]\otimes[K_4,K_4].
\]
This simply records the two-step interval from a codimension-one boundary
to the full pentagon.

More interestingly, for a vertex $V_i$ there are two intermediate facets,
namely the two edges of the pentagon meeting at that vertex.  Thus, if
$V_i=E_a\cap E_b$, then
\[
\begin{aligned}
   \Delta[V_i,K_4]
   ={}&[V_i,V_i]\otimes[V_i,K_4]
      +[V_i,E_a]\otimes[E_a,K_4]  \\
     &+[V_i,E_b]\otimes[E_b,K_4]
      +[V_i,K_4]\otimes[K_4,K_4].
\end{aligned}
\]
For example,
\[
\begin{aligned}
   \Delta[V_1,K_4]
   ={}&[V_1,V_1]\otimes[V_1,K_4]
      +[V_1,E_{13}]\otimes[E_{13},K_4] \\
     &+[V_1,E_{14}]\otimes[E_{14},K_4]
      +[V_1,K_4]\otimes[K_4,K_4].
\end{aligned}
\]

This formula has a direct amplitude interpretation.  The vertex $V_1$
corresponds to the term
\[
   \frac{1}{X_{13}X_{14}}.
\]
There are two possible first residues leading to this vertex: one may
first take the residue on the facet $X_{13}=0$ and then go to the corner
$X_{14}=0$, or first take the residue on the facet $X_{14}=0$ and then go
to the corner $X_{13}=0$.  The two middle terms in
$\Delta[V_1,K_4]$ record precisely these two possible orders of boundary
factorisation.

Thus, already for $K_4$, the boundary--residue incidence coalgebra does
more than list the faces of the polytope.  It records the nested ways in
which the canonical form can approach a given corner through intermediate
factorisation channels.  This is the finite-dimensional positive-geometric
analogue of the way in which the Connes--Kreimer coproduct records nested
subdivergence structures in a Feynman graph.  In the present tree-level
case the relevant recursion is not renormalisation, but residue
factorisation.

\begin{remark}
The construction above should not be interpreted as a renormalisation
Hopf algebra. At tree level there are no ultraviolet subdivergences and
no counterterms. Rather, the boundary--residue incidence coalgebra is the
positive-geometric analogue of the recursive combinatorial structure
seen in the Connes--Kreimer Hopf algebra: divergent subgraphs are
replaced by boundary faces, contracted cographs by face intervals, and
renormalisation recursion by residue recursion.
\end{remark}

\subsection{Loop corrections and logarithmic singularity strata}
\label{subsec:loop-scope}

The associahedral theorem above is a tree-level statement, but its proof
uses only two ingredients: a finite boundary poset and the recursive
residue property of the canonical form.  This makes the construction
portable.  The associahedron should not be used at loop level merely by
analogy; rather, when a loop-level positive geometry is known, its own
boundary stratification supplies the relevant poset.

\begin{definition}[Loop boundary--residue incidence coalgebra]
Let $\mathcal P_{n,L}$ be a loop-level positive geometry, or a positive
cellular model for a loop integrand, with finite boundary poset
$\mathcal B(\mathcal P_{n,L})$.  Let $F\leq G$ denote inclusion of
boundary strata.  Define
\[
   \mathcal C(\mathcal P_{n,L})
   :=
   \mathbb Q\bigl\langle [F,G] : F\leq G,
   \; F,G\in \mathcal B(\mathcal P_{n,L})\bigr\rangle
\]
with coproduct
\[
   \Delta [F,G]
   =
   \sum_{F\leq H\leq G}[F,H]\otimes[H,G]
\]
and counit
\[
   \varepsilon([F,G])=
   \begin{cases}
   1, & F=G,\\
   0, & F<G.
   \end{cases}
\]
We call $\mathcal C(\mathcal P_{n,L})$ the loop boundary--residue
incidence coalgebra of $\mathcal P_{n,L}$.
\end{definition}

\begin{proposition}
The coproduct on $\mathcal C(\mathcal P_{n,L})$ is coassociative.  If the
canonical form $\Omega_{\mathcal P_{n,L}}$ has boundary residues equal, up
to orientation, to the canonical forms of boundary positive geometries,
then the coproduct records all possible intermediate residue
factorisations of $\Omega_{\mathcal P_{n,L}}$.
\end{proposition}

\begin{proof}
Coassociativity is the same interval-poset calculation used above:
\[
   (\Delta\otimes\mathrm{id})\Delta[F,G]
   =
   \sum_{F\leq H\leq K\leq G}
   [F,H]\otimes[H,K]\otimes[K,G]
   =
   (\mathrm{id}\otimes\Delta)\Delta[F,G].
\]
If $F\leq H\leq G$, then the corresponding term
$[F,H]\otimes[H,G]$ records the two-stage residue process in which one
first approaches the intermediate boundary stratum $H$ and then the lower
stratum $F$.  Since the defining property of a positive geometry is that
residues of the canonical form on boundary components are the canonical
forms of those boundary components, the incidence coproduct is the
combinatorial record of all such nested residue decompositions.
\end{proof}

For planar loop integrands this gives a concrete prescription.  One
replaces the tree-level associahedral face poset by the boundary poset of
the relevant loop positive geometry.  In planar $\mathcal N=4$
super-Yang--Mills theory, the loop integrand is described by the positive
Grassmannian/amplituhedron formalism, with $L$ loop variables represented
by $L$ lines in momentum-twistor space
\cite{ArkaniHamedBourjailyCachazoCaronHuotTrnka2011,ArkaniHamedPositiveGrassmannian2016,ArkaniHamedTrnka2014,ArkaniHamedHodgesTrnka2014}.  The corresponding coalgebra is therefore not the associahedral coalgebra
$\mathcal C(K_n)$, but the incidence coalgebra of the boundary
stratification of the loop amplituhedron, or of a chosen positive cell
decomposition computing the same integrand.

For planar scalar amplitudes there is an even closer one-loop analogue of
the associahedral story.  The halohedron gives a positive geometry for the
planar one-loop biadjoint scalar integrand
\cite{Salvatori2018,SalvatoriCacciatori2018}.  If $H_n$ denotes the
halohedron associated with the $n$-point one-loop planar scalar amplitude,
we define
\[
   \Delta [F,H_n]
   =
   \sum_{F\leq G\leq H_n}[F,G]\otimes[G,H_n],
\]
where the sum ranges over the boundary strata $G$ of the halohedron
containing $F$.  Boundary facets of the halohedron include the loop-level
degenerations and factorisation channels appropriate to the one-loop
integrand; hence this coproduct records nested one-loop planar scalar
residue channels in exactly the same incidence-coalgebraic sense in which
Theorem~\ref{thm:Kn-boundary-residue} records nested tree-level planar
factorisations.

The non-planar case requires a separate formulation.  There is evidence
that non-planar maximally supersymmetric amplitudes retain logarithmic
singularity and leading-singularity structures familiar from the planar
story
\cite{ArkaniHamedBourjailyCachazoTrnka2014,BernHerrmannLitseyStankowiczTrnka2015,BernHerrmannLitseyStankowiczTrnka2016}.  However, no universal non-planar positive
geometry playing the role of the planar amplituhedron is presently known.
For this reason one can define, without assuming such a geometry, a
logarithmic-singularity incidence coalgebra as follows.  Let
$\mathcal S_{n,L}^{\mathrm{np}}$ be any finite poset of logarithmic
singularity strata of a non-planar $L$-loop integrand, ordered by inclusion
or by specialization of singular loci.  Set
\[
   \mathcal C_{\log}^{\mathrm{np}}(n,L)
   :=
   \mathbb Q\bigl\langle [S,T]: S\leq T,
   \; S,T\in\mathcal S_{n,L}^{\mathrm{np}}\bigr\rangle,
\]
with
\[
   \Delta[S,T]=\sum_{S\leq U\leq T}[S,U]\otimes[U,T].
\]
This is an incidence coalgebra whether or not the singularity poset is
known to come from the boundary stratification of a positive geometry.  If
a non-planar positive geometry is eventually identified, this logarithmic
singularity coalgebra should be recovered from its boundary poset.  Until
then it gives a precise algebraic language for organizing nested
non-planar logarithmic singularities.

This gives a precise scope for loop corrections.  For planar loop
integrands, one should replace the associahedral face poset by the
boundary poset of the corresponding loop positive geometry.  For one-loop
planar scalar amplitudes, the halohedron gives a concrete candidate.  For
non-planar amplitudes, the same coalgebraic language may organise known
logarithmic singularity structures, but the relevant positive geometry
remains part of the open problem.

\section{Cellular Reconstruction and Incidence Coalgebras}
\label{sec:cellular-reconstruction}

The constructions above suggest a broader mathematical pattern.  The
boundary--residue coalgebra of a positive geometry is not merely an ad hoc
device for associahedra.  It is an incidence coalgebra attached to a
stratified space whose strata carry physical residue data.  The same
incidence mechanism already occurs in algebraic topology, where spaces are
built from cells and the topology of the resulting space is encoded by the
way in which the cells are incident to one another.  This section makes that
comparison precise.

The aim is not to identify spacetime with an associahedron, nor to claim that
a Lorentzian metric is determined by a face poset.  The point is more basic:
cellular models of spacetime and positive geometries of amplitudes both have
combinatorial skeletons, and both skeletons naturally define incidence
coalgebras.  On the spacetime side the incidence data organise cellular
reconstruction and homology.  On the amplitude side the incidence data
organise canonical-form residues and factorisation channels.

\subsection{Skeleta, face posets and reconstruction}

Let $X$ be a finite regular CW complex.  We write $X^{(k)}$ for its
$k$-skeleton, namely the union of all cells of dimension at most $k$.  Thus

a finite CW complex is built inductively by attaching $k$-cells to the
$(k-1)$-skeleton.  Regularity means that the closure of each open cell is a
closed ball and that its boundary is a union of lower-dimensional cells.  This
is the setting in which the face-poset language behaves most cleanly
\cite{Whitehead1949,Hatcher2002,Munkres1984,Bjorner1984}.

Let $\mathcal F(X)$ denote the poset of non-empty closed cells of $X$, ordered
by inclusion.  A chain
\[
   \sigma_0 < \sigma_1 < \cdots < \sigma_r
\]
in $\mathcal F(X)$ is a flag of cells.  The order complex
$\Delta(\mathcal F(X))$ is the simplicial complex whose simplices are precisely
such chains.  In the simplicial case this is the usual barycentric
subdivision; for finite regular CW complexes it is the barycentric
subdivision in the corresponding cellular sense.

\begin{proposition}[Combinatorial reconstruction]
\label{prop:face-poset-reconstruction}
Let $X$ be a finite regular CW complex and let $\mathcal F(X)$ be its face
poset of non-empty closed cells.  Then the order complex
$\Delta(\mathcal F(X))$ is naturally identified with the barycentric
subdivision of $X$.  Consequently its geometric realisation is homeomorphic
to $X$:
\[
   |\Delta(\mathcal F(X))| \cong X .
\]
In particular, for a finite simplicial complex $K$, the face poset of $K$
determines the underlying polyhedron $|K|$ up to homeomorphism.
\end{proposition}

\begin{proof}
In a simplicial complex, a simplex of the barycentric subdivision is a chain of
non-empty faces of an original simplex.  Such chains are exactly the simplices
of the order complex of the face poset.  Hence
$\Delta(\mathcal F(K))$ is the barycentric subdivision of $K$.

For a finite regular CW complex the same argument applies cell by cell.  The
regularity assumption ensures that the closure of each cell is a ball and that
its boundary is a union of lower-dimensional cells.  Choosing the barycentre
of each cell, chains of incident cells give the simplices of the barycentric
subdivision.  These chains are precisely the simplices of
$\Delta(\mathcal F(X))$.  Barycentric subdivision does not change the
underlying topological realisation, and the result follows.
\end{proof}

This proposition is the precise topological form of the slogan that a space
can be reconstructed from its cellular skeleton.  The statement is not that a
metric, a smooth structure, or a causal structure is determined by the face
poset.  The statement is that, under the regularity hypotheses above, the
combinatorics of incidence reconstructs the underlying topological space.
This is the level at which the comparison with positive geometry becomes
mathematically sharp.

\subsection{The cellular incidence coalgebra}

The same interval coalgebra used for positive geometries can now be attached
to a cellular space.

\begin{definition}[Cellular incidence coalgebra]
Let $X$ be a finite regular CW complex with face poset $\mathcal F(X)$.  The
cellular incidence coalgebra of $X$ is the vector space
\[
   \mathcal C_{\mathrm{cell}}(X)
   :=
   \mathbb Q\bigl\langle [\sigma,\tau] :
   \sigma\leq \tau,
   \; \sigma,\tau\in\mathcal F(X)\bigr\rangle
\]
with coproduct
\[
   \Delta[\sigma,\tau]
   =
   \sum_{\sigma\leq \rho\leq \tau}
   [\sigma,\rho]\otimes[\rho,\tau]
\]
and counit
\[
   \varepsilon([\sigma,\tau])=
   \begin{cases}
   1, & \sigma=\tau,\\
   0, & \sigma<\tau.
   \end{cases}
\]
\end{definition}

\begin{proposition}
The cellular incidence coalgebra $\mathcal C_{\mathrm{cell}}(X)$ is
coassociative.  Its iterated coproduct records all flags of cells between two
given cells $\sigma\leq\tau$.
\end{proposition}

\begin{proof}
The proof is identical to the interval-poset calculation already used for
positive geometries.  For any interval $[\sigma,\tau]$,
\[
   (\Delta\otimes\mathrm{id})\Delta[\sigma,\tau]
   =
   \sum_{\sigma\leq\rho\leq\eta\leq\tau}
   [\sigma,\rho]\otimes[\rho,\eta]\otimes[\eta,\tau]
\]
and
\[
   (\mathrm{id}\otimes\Delta)\Delta[\sigma,\tau]
   =
   \sum_{\sigma\leq\rho\leq\eta\leq\tau}
   [\sigma,\rho]\otimes[\rho,\eta]\otimes[\eta,\tau].
\]
The two sums are the same.  Iterating $\Delta$ inserts longer chains of
intermediate cells between $\sigma$ and $\tau$, and therefore records flags of
cells.
\end{proof}

To recover the usual cellular chain complex one must add orientation data.
If $X$ is an oriented finite regular CW complex, the cellular boundary map is
\[
   \partial \tau
   =
   \sum_{\substack{\sigma<\tau\\ \dim\sigma=\dim\tau-1}}
   [\tau:\sigma] \, \sigma,
\]
where $[\tau:\sigma]\in\{0,\pm1\}$, or more generally an integer incidence
number, records the degree with which the oriented boundary of $\tau$ meets
$\sigma$.

\begin{proposition}[Rank-one incidence and cellular boundary]
\label{prop:rank-one-boundary}
Let $X$ be an oriented finite regular CW complex.  The rank-one intervals
$[\sigma,\tau]$ with $\dim\tau=\dim\sigma+1$, together with the orientation
incidence numbers $[\tau:\sigma]$, determine the cellular boundary operator.
Consequently the incidence coalgebra supplies the unsigned flag structure,
while the oriented incidence numbers supply the signs needed for cellular
homology.
\end{proposition}

\begin{proof}
The cellular boundary of a cell $\tau$ is computed by summing over precisely
those cells $\sigma$ of codimension one lying in the boundary of $\tau$.
These are exactly the rank-one intervals $[\sigma,\tau]$ in the face poset.
The incidence coalgebra records the existence of these rank-one intervals and
of all longer flags containing them.  The additional orientation numbers
$[\tau:\sigma]$ determine the signed coefficients in the cellular boundary
map.  Hence the rank-one part of the incidence structure, decorated by
orientation data, is equivalent to the cellular boundary operator.
\end{proof}

The comparison with positive geometry is now direct.  For a positive geometry
$P$, the unsigned interval $[F,G]$ records a possible passage from a boundary
stratum $G$ to a deeper boundary stratum $F$.  Orientation conventions for
canonical forms determine signs of residues.  For a CW complex, the interval
$[\sigma,\tau]$ records a possible passage from a cell to one of its boundary
cells, while incidence numbers determine the signs in the cellular boundary
map.  In both cases the primitive structure is a poset of boundary relations
plus additional sign or orientation data.

\begin{example}[A two-simplex]
Let $\Delta^2$ be an oriented triangle with vertices $v_0,v_1,v_2$, edges
$e_{01},e_{12},e_{02}$ and two-cell $f$.  The face poset contains intervals
such as $[v_0,e_{01}]$, $[e_{01},f]$ and $[v_0,f]$.  The coproduct of the last
interval is
\[
   \Delta[v_0,f]
   =
   [v_0,v_0]\otimes[v_0,f]
   +[v_0,e_{01}]\otimes[e_{01},f]
   +[v_0,e_{02}]\otimes[e_{02},f]
   +[v_0,f]\otimes[f,f].
\]
Thus the coalgebra records the two possible one-edge routes from the vertex
$v_0$ to the face $f$.  The oriented cellular boundary
\[
   \partial f = e_{01}+e_{12}-e_{02}
\]
requires the additional incidence signs.  This is exactly analogous to the
positive-geometric situation: the face poset records possible residue flags,
whereas the canonical form and orientation conventions determine the signed
residue.
\end{example}

\subsection{Spacetime as a labelled cellular object}

In classical general relativity, spacetime is modelled by a smooth Lorentzian
manifold.  A smooth manifold can often be triangulated or given a handle or CW
decomposition, and in the PL category one may work directly with a finite or
locally finite cellular model.  The topological content of such a model is
combinatorial, but the physical spacetime contains more data than topology:
it carries metric, causal and differentiable structure.

This distinction is essential.  The face poset of a triangulation determines
the underlying polyhedron after barycentric subdivision, but it does not
determine edge lengths, deficit angles, light cones or curvature.  These must
be added as labels or additional structures.  Regge calculus is the classical
model of this idea: the smooth metric is replaced by a piecewise-flat
simplicial geometry, and curvature is concentrated on codimension-two hinges
through deficit angles \cite{Regge1961}.  Causal dynamical triangulations add
Lorentzian causal restrictions to the triangulated histories
\cite{AmbjornJurkiewiczLoll2001}.  Causal set theory takes the causal order as
the primitive discrete structure, with local finiteness providing volume
information \cite{BombelliLeeMeyerSorkin1987}.  Spin-network and spin-foam
models decorate graphs and two-complexes by representation-theoretic labels,
thereby turning combinatorial data into quantum-geometric data
\cite{RovelliSmolin1995,Rovelli2004,AshtekarLewandowski2004}.

These examples show that the passage from combinatorics to spacetime is never
only a matter of an abstract poset.  One needs labels.  Nevertheless, the
incidence skeleton is the common underlying object.  We may therefore isolate
the following structure.

\begin{definition}[Labelled spacetime incidence structure]
A labelled cellular spacetime model consists of a finite or locally finite
regular CW complex $X$, its face poset $\mathcal F(X)$, and a system of labels
$\lambda$ on cells or incidences.  Depending on the model, the labels may be
metric data, causal order data, representation-theoretic data, or weights in a
state sum.  The associated incidence coalgebra is the interval coalgebra of
$\mathcal F(X)$, together with the label system $\lambda$.
\end{definition}

For example, in Regge calculus the labels include squared edge lengths; in a
causal set the primitive relation is a partial order and the counting measure
is part of the geometry; in spin-foam models the labels include group
representations and intertwiners.  The incidence coalgebra does not replace
these labels.  It provides the common combinatorial backbone on which they
live.

\subsection{The common lower row}

The positive-geometric construction and the cellular-spacetime construction
can now be displayed as parallel passages through incidence data.  On the
amplitude side one has
\[
\begin{array}{ccc}
\text{positive geometry} & \longrightarrow & \text{canonical form / amplitude} \\
\downarrow & & \downarrow \\
\text{face poset} & \longrightarrow & \text{boundary--residue coalgebra},
\end{array}
\]
where the right vertical arrow is implemented by taking residues of the
canonical form.  On the cellular spacetime side one has
\[
\begin{array}{ccc}
\text{cellular spacetime} & \longrightarrow & \text{cellular chain complex / geometry} \\
\downarrow & & \downarrow \\
\text{face poset} & \longrightarrow & \text{cellular incidence coalgebra},
\end{array}
\]
where the right vertical arrow is implemented by attaching maps, incidence
numbers, and any additional metric or causal labels.

The lower row is the same in both diagrams.  In each case, one passes from a
geometric object to a poset of boundary relations and then to an incidence
coalgebra.  What differs is the extra structure placed on that coalgebra.  In
cellular topology the extra structure gives boundary operators, homology and,
with further labels, metric or causal geometry.  In positive geometry the
extra structure is the canonical form, whose logarithmic singularities and
residues produce factorisation.

\begin{proposition}[Common incidence skeleton]
Let $P$ be a positive geometry with finite face poset, and let $X$ be a finite
regular CW complex.  The coalgebras $\mathcal C(P)$ and
$\mathcal C_{\mathrm{cell}}(X)$ are both instances of the same interval
coalgebra construction applied to a finite poset.  In the first case,
intervals encode nested boundary residues of canonical forms; in the second,
intervals encode nested cell incidences and, after adding orientation data,
cellular boundaries.
\end{proposition}

\begin{proof}
Both coalgebras are defined on the vector space generated by intervals of a
finite poset.  The coproduct in each case is
\[
   \Delta[a,b]=\sum_{a\leq c\leq b}[a,c]\otimes[c,b].
\]
The algebraic structure is therefore identical.  The difference lies in the
interpretation of the underlying poset.  For $P$, the poset is the boundary
stratification of a positive geometry and the intervals are interpreted by
residue operations on canonical forms.  For $X$, the poset is the face poset
of cells and the intervals are interpreted by cellular incidence relations.
Thus the same interval-coalgebra skeleton supports two different geometric
interpretations.
\end{proof}

\subsection{A functorial formulation}

The comparison above can be expressed in a functorial language in the sense of
standard category theory \cite{MacLane1998}.  This is not only a matter of
notation: it clarifies exactly which part of the construction is common to
cellular topology and to positive geometry, and which part is additional
structure.

Let \(\mathbf{FinPos}_{\mathrm{int}}\) denote the category whose objects are
finite posets and whose morphisms are \emph{interval embeddings}.  Thus a
morphism \(\varphi:P\to Q\) is an injective order-preserving map with the
property that, for every interval \([a,b]\subset P\), the induced map
\[
   [a,b]\longrightarrow [\varphi(a),\varphi(b)]
\]
is an isomorphism of posets.  Equivalently, every element of \(Q\) lying
between \(\varphi(a)\) and \(\varphi(b)\) is the image of a unique element of
\([a,b]\).  This restriction is imposed because a general monotone map of
posets need not preserve interval coalgebras.

Define a functor
\[
   \mathcal I:
   \mathbf{FinPos}_{\mathrm{int}}
   \longrightarrow
   \mathbf{Coalg}_{\mathbb Q}
\]
by assigning to a finite poset \(P\) the interval coalgebra
\[
   \mathcal I(P)
   =
   \mathbb Q\langle [a,b]:a\leq b\text{ in }P\rangle,
\]
with coproduct
\[
   \Delta[a,b]=\sum_{a\leq c\leq b}[a,c]\otimes[c,b].
\]
If \(\varphi:P\to Q\) is an interval embedding, set
\[
   \mathcal I(\varphi)([a,b])=[\varphi(a),\varphi(b)].
\]

\begin{proposition}[Functoriality of the interval coalgebra]
\label{prop:interval-functor}
The assignment \(\mathcal I\) is a functor from
\(\mathbf{FinPos}_{\mathrm{int}}\) to coalgebras over \(\mathbb Q\).
\end{proposition}

\begin{proof}
It is immediate that identities and compositions are preserved.  It remains
only to check compatibility with the coproduct.  For an interval
\([a,b]\subset P\),
\[
   \Delta \mathcal I(\varphi)([a,b])
   =
   \Delta[\varphi(a),\varphi(b)]
   =
   \sum_{\varphi(a)\leq q\leq\varphi(b)}
   [\varphi(a),q]\otimes[q,\varphi(b)].
\]
Since \(\varphi\) is an interval embedding, each such \(q\) is uniquely of the
form \(q=\varphi(c)\), with \(a\leq c\leq b\).  Hence
\[
   \Delta \mathcal I(\varphi)([a,b])
   =
   \sum_{a\leq c\leq b}
   [\varphi(a),\varphi(c)]\otimes[\varphi(c),\varphi(b)]
   =
   (\mathcal I(\varphi)\otimes\mathcal I(\varphi))\Delta[a,b].
\]
The counit is preserved for the same reason: \(a=b\) if and only if
\(\varphi(a)=\varphi(b)\).  Thus \(\mathcal I(\varphi)\) is a coalgebra
morphism.
\end{proof}

Now let \(\mathbf{PosGeom}_{\mathrm{fin}}^{\mathrm{int}}\) be a category of
finite positive geometries with morphisms restricted to boundary-preserving
maps that induce interval embeddings on face posets.  Similarly, let
\(\mathbf{RegCW}_{\mathrm{fin}}^{\mathrm{int}}\) be the category of finite
regular CW complexes with cellular morphisms that induce interval embeddings
on face posets.  There are face-poset functors
\[
   \mathcal F_{+}:
   \mathbf{PosGeom}_{\mathrm{fin}}^{\mathrm{int}}
   \longrightarrow
   \mathbf{FinPos}_{\mathrm{int}},
   \qquad
   P\longmapsto\mathcal F(P),
\]
and
\[
   \mathcal F_{\mathrm{cell}}:
   \mathbf{RegCW}_{\mathrm{fin}}^{\mathrm{int}}
   \longrightarrow
   \mathbf{FinPos}_{\mathrm{int}},
   \qquad
   X\longmapsto\mathcal F(X).
\]
Composing with \(\mathcal I\) gives two coalgebra-valued functors:
\[
   \mathcal I\circ\mathcal F_{+}
   :
   \mathbf{PosGeom}_{\mathrm{fin}}^{\mathrm{int}}
   \longrightarrow
   \mathbf{Coalg}_{\mathbb Q},
\]
and
\[
   \mathcal I\circ\mathcal F_{\mathrm{cell}}
   :
   \mathbf{RegCW}_{\mathrm{fin}}^{\mathrm{int}}
   \longrightarrow
   \mathbf{Coalg}_{\mathbb Q}.
\]
Thus the two constructions literally factor through the same interval-coalgebra
functor:
\[
\begin{array}{ccccc}
\mathbf{PosGeom}_{\mathrm{fin}}^{\mathrm{int}}
& \xrightarrow{\ \mathcal F_{+}\ }
& \mathbf{FinPos}_{\mathrm{int}}
& \xrightarrow{\ \mathcal I\ }
& \mathbf{Coalg}_{\mathbb Q},
\end{array}
\]
and
\[
\begin{array}{ccccc}
\mathbf{RegCW}_{\mathrm{fin}}^{\mathrm{int}}
& \xrightarrow{\ \mathcal F_{\mathrm{cell}}\ }
& \mathbf{FinPos}_{\mathrm{int}}
& \xrightarrow{\ \mathcal I\ }
& \mathbf{Coalg}_{\mathbb Q}.
\end{array}
\]

This is the precise functorial content of the common lower row.  The shared
object is not a metric spacetime and not an amplitude; it is the interval
coalgebra of a finite incidence poset.  The two theories differ by the extra
data placed over this common functorial skeleton.  On the positive-geometric
side one adds canonical forms and residue maps.  On the cellular-spacetime
side one adds attaching maps, orientation numbers and, if the cellular complex
is to model physical spacetime, metric, causal or quantum-geometric labels.
The functorial viewpoint is therefore deliberately modest, but it isolates the
mathematical mechanism common to the two reconstructions.

\subsection{An incidence-first reconstruction principle}

The comparison suggests the following principle.  It is not a theorem of
quantum gravity, but a precise organising principle connecting the present
positive-geometric construction to the cellular viewpoint of algebraic
topology.

\begin{center}
\fbox{\begin{minipage}{0.88\textwidth}
\textbf{Incidence-first reconstruction principle.}
The primitive object common to cellular spacetime models and positive-geometric
amplitudes is an incidence structure.  In the spacetime case, geometric
realisation of this incidence structure reconstructs the underlying
polyhedron, once suitable regularity assumptions are imposed, while additional
metric or causal labels encode geometry.  In the amplitude case, the same
incidence structure, equipped with canonical forms, reconstructs the nested
factorisation structure of scattering amplitudes through residues.
\end{minipage}}
\end{center}

This principle gives a controlled sense in which spacetime-like structure and
amplitude factorisation may both be governed by combinatorial boundary data.
It also clarifies the limitation.  Incidence data alone reconstruct topology,
not physics.  To obtain a physical spacetime one adds metric, causal or
quantum-geometric labels.  To obtain an amplitude one adds the canonical form
and its residue prescription.  The point of the comparison is that both kinds
of reconstruction pass through the same coalgebraic skeleton.

There is one further technical caveat.  Not every topological manifold in all
dimensions is triangulable, and the triangulation problem has subtle high-
dimensional obstructions \cite{Manolescu2016}.  The statement made here is
therefore deliberately restricted to triangulated, PL, smooth or regular CW
models of the kind that occur in most physical applications.  Within that
scope, the analogy is precise: the face poset reconstructs the cellular
realisation, and the interval coalgebra records the nested incidence data.

This provides the promised bridge.  The boundary--residue incidence coalgebra
introduced for positive geometries is the amplitude analogue of the cellular
incidence coalgebra of a decomposed spacetime.  In algebraic topology,
cellular incidence organises the construction of space.  In positive geometry,
boundary incidence organises the construction of factorisation.  The same
combinatorial mechanism lies beneath both descriptions.

\section{Concluding Remarks}
\label{sec:conclusion}

The paper has introduced a boundary--residue incidence coalgebra associated
with the finite face poset of a positive geometry.  The construction is the
ordinary incidence coalgebra of intervals in a poset, but it acquires a
specific amplitude-theoretic interpretation when the poset is the boundary
stratification of a positive geometry and the residue recursion of the
canonical form is taken into account.

For the associahedron $K_n$, faces are indexed by admissible sets of
non-crossing diagonals of a convex $(n+1)$-gon.  We proved that the interval
coproduct from a face $F_D$ to the full associahedron is
\[
   \Delta [F_D,K_n]
   =
   \sum_{E\subseteq D}[F_D,F_E]\otimes[F_E,K_n],
\]
where the sum runs over all sub-dissections $E$ of $D$.  This formula records
all possible intermediate choices of planar factorisation channels.  The
corresponding residue of the canonical form factorises as an exterior product
of canonical forms on the lower associahedra associated with the subpolygons
obtained by cutting along $D$.

The worked example of $K_4$ makes the construction explicit in the first
non-trivial case.  The two middle terms in the coproduct from a vertex to the
pentagon record the two possible first residues leading to the corresponding
corner of the canonical form.  This is the tree-level positive-geometric
analogue of the way in which the Connes--Kreimer coproduct records nested
substructures of Feynman graphs, with residue factorisation replacing
renormalisation recursion.

The result should therefore be read first as a precise tree-level statement
about associahedral positive geometries.  It identifies the incidence
coproduct of face intervals as the combinatorial shadow of residue
factorisation of the canonical form.  The loop-level formulation shows that
this mechanism is not tied to the associahedron itself: when a planar loop
positive geometry is available, the same construction is applied to its
boundary poset.  For non-planar amplitudes, where the corresponding positive
geometry is not known in comparable generality, the logarithmic-singularity
incidence coalgebra isolates the algebraic structure that such a geometry
would have to realise.

The final cellular comparison places the construction in a broader
topological context.  The face poset of a triangulated or regular CW space
reconstructs its barycentric subdivision, while the corresponding interval
coalgebra records flags of cells.  In positive geometry the same interval
coalgebra records flags of boundary strata and hence iterated residues.  The
functorial formulation makes this commonality precise: both the positive-
geometric and cellular-spacetime constructions factor through the same
face-poset functor followed by the same interval-coalgebra functor.  They
differ in the structures placed over that common coalgebraic skeleton:
spacetime models require metric, causal or representation-theoretic labels,
while amplitudes require canonical forms and residue prescriptions.  This is
the level at which the comparison with both Connes--Kreimer and cellular
topology is made in the present paper: nested subdivergences are replaced by
nested boundary strata or singularity strata, and renormalisation recursion is
replaced by residue recursion.

\end{document}